\documentclass[12pt]{article}
\usepackage{amssymb,amsmath,amsthm,latexsym,graphicx,color,cite}
\newtheorem{theorem}{Theorem}
\newtheorem{thm}{Theorem}[section]
\newtheorem{prop}[thm]{Proposition}
\newtheorem{lemma}[thm]{Lemma}

\newtheorem{remark}[thm]{Remark}

\newtheorem{defn}[thm]{Definition}

\newcommand{\Go}{G_0}

\newcommand{\g}{\gamma}
\newcommand{\G}{\Gamma}
\newcommand{\dn}{d\,'}
\newcommand{\fh}{\hat{f}}

\newcommand{\R}{\mathbb R}
\newcommand{\RR}{\mathbb R}
\newcommand{\Z}{\mathbb Z}

\newcommand{\C}{\mathbb C}

\DeclareMathOperator{\supp}{supp}

\newcommand{\p}{\partial}

\newcommand{\vp}{\varphi}
\newcommand{\atopp}[2]{\genfrac{}{}{0pt}{2}{#1}{#2}}

\newcommand{\lam}{\lambda}

\newcommand{\loc}{\text{loc}}
\newcommand{\T}{{\mathbb{T}}}

\date{\small Mathematics Subject Classification: 35J08, 35J15, 35P05, 47A10, 47F05, 47N20}

\begin{document}

\title{Green's function asymptotics near the internal edges of spectra of periodic elliptic operators. Spectral edge case.}

\author{Peter Kuchment\footnote{Mathematics Department, Texas A\&M University,
College Station, TX 77843-3368, USA. kuchment@math.tamu.edu} and Andrew Raich\footnote{Department of Mathematical Sciences, University of Arkansas,
Fayetteville, AR 72701, USA. araich@uark.edu}\\ \emph{Dedicated to the  75th birthday of Professor Eduard Tsekanovskii.}}

\maketitle

\begin{abstract}
Precise asymptotics known for the Green's function of the Laplace operator have found their analogs for periodic elliptic operators of the second order at and below the bottom of the spectrum. Due to the band-gap structure of the spectra of such operators, the question arises whether similar results can be obtained near or at the edges of spectral gaps.
As the result of this work shows, this is possible at a spectral edge when the dimension $d\geq 3$.
\end{abstract}

\date{}
\section{Introduction}

Let
\begin{equation}\label{eqn:L definition}
L(x,D) = \sum_{j,\ell=1}^d \left(D_j +b_j(x)\right)a_{j l}(x) \left(D_l +b_l(x)\right) + c(x),
\end{equation}
be a second-order differential operator in $\R^d$ with smooth real valued coefficients $a_{j l}, b_j, c$. Here $D=(D_1,\dots,D_d)$ and $D_j = -i\p/\p x_j$.

The operator $L$ is assumed to be \textbf{elliptic}, i.e., the matrix $a_{j\ell}$ is symmetric and
\begin{equation}\label{E:ellipt}
\sum_{j,\ell=1}^d a_{j\ell}\xi_j\xi_\ell \geq \theta |\xi|^2
\end{equation}
for some $\theta>0$ and any $\xi = (\xi_1,\dots,\xi_d)\in \R^d$.

We assume that all coefficients are \textbf{$\Z^d$-periodic} where $\Z^d$ is the integer lattice in $\R^d$, i.e., $a_{j l}(x+n) = a_{j l}(x)$ for all $x\in\R^d, n\in\Z^d$ and similarly for $b_j$ and $c$.

If the operator is defined as in (\ref{eqn:L definition}) on the subspace $C_0^\infty(\R^d)\subset L^2(\R^d)$ of compactly supported smooth functions, it is symmetric and essentially self-adjoint \cite{Shubin}. The corresponding self-adjoint operator $L$ has the domain $H^2(\R^d)$.

Most of the conditions above (second order, smoothness of the coefficients, lattice of periods being equal to $\Z^d$, and sometimes even self-adjointness) can be relaxed,
and we postpone a further discussion until Section \ref{S:remarks}.

In order to formulate the problem addressed in this text and state the results, we  introduce some notation
and notions from the spectral theory of periodic elliptic operators (see, e.g. \cite{Kuch_book,ReedSimon,Eastham}).

The spectrum of  the defined above (see \eqref{eqn:L definition}) operator $L$ in $L^2(\R^d)$ has a \textbf{band-gap structure}, i.e., it is the union of a sequence of closed bounded intervals (\textbf{bands}) $[\alpha_j,\beta_j]\subset \R,j=1,2,\dots$ such that $a_j,b_j \mathop{\to}\limits_{j\to\infty}\infty$:
\begin{equation}\label{E:spectrum}
    \sigma(L)=\bigcup\limits_j [\alpha_j,\beta_j].
\end{equation}
The bands can (and do, when $d>1$ \cite{Parn,ParnSob,Skrig}) overlap, but they may leave open intervals in between, called \textbf{spectral gaps}.
Existence of such gaps, in particular, is crucial for the properties of semi-conductors \cite{AM} and photonic crystals \cite{Kuch_pbg,JMW}.

Our interest here is in obtaining asymptotics of the Green's function for the operator $L$ near and at the boundaries of its spectrum. These asymptotics are relevant for many problems
in theory of random walks, Martin boundaries, Anderson localization, and others (e.g., \cite{AizMolch,Murata2,Murata6,Babillot1,BCH}). This has been done for the case of the points at and below the bottom of the spectrum of $L$ \cite{Babillot1,Murata2,Murata6} (see also \cite{Woess} for description of the results of \cite{Babillot1} in discrete case) and at the internal edges in one-dimensional case \cite{Murata6}. The resulting asymptotics are similar to the ones for the Laplace operator $-\Delta$, except the appearance of oscillatory terms. While $\lambda=0$ is the only boundary point of the spectrum of the Laplace operator,
spectral gaps may appear in the periodic case, and thus the question arises of whether one can obtain precise information about asymptotics of the Green's function inside these gaps, or at their edges. The Combes-Thomas estimates (see \cite{BCH} for their contemporary version) provide the decay estimates for the resolvent in terms of the distance to the spectrum. They, however, do not reflect the exact asymptotic behavior of the Green's function. It is the goal of this and the next paper to obtain  asymptotics or (at least) estimates at and near the spectral gap edges. In this text, only the case of a spectral edge is addressed in dimension $d\geq 3$, while the asymptotics of the Green's function inside a spectral gap (near its edge) will be considered elsewhere. The restriction $d\geq 3$ in the spectral edge case is natural, since at the bottom of the spectrum, the operator is critical in dimensions $d\leq 2$ (see \cite{Murata2,Pinsky,LinPinch}).

Let $W=[0,1]^d\subset \R^d$ be the unit cube. $W$ is a \textbf{fundamental domain} of $\R^d$ with respect to the lattice $\G:=\Z^d$, i.e., $W$ and its integer translates cover $\R^d$
and overlap only on the boundaries.

The \textbf{dual} (or \textbf{reciprocal}) \textbf{lattice}\footnote{For a lattice $\G$, the elements $\kappa$ of the dual lattice are characterized by the property that $\kappa\cdot\g\in2\pi\Z$ for any $\g\in\G$.}
of $\G$ is $\G^*:=2\pi\Z^d$  and its fundamental domain is $B=[-\pi,\pi]^d$.
We call $B$ the \textbf{Brillouin zone} (although in the solid state physics usually a somewhat different fundamental domain of the reciprocal lattice bears this name \cite{AM}).

We  also use the notation $\T = \R^d/\Z^d$ and $\T^* = \R^d/2\pi\Z^d$  for the $d$-dimensional tori corresponding to the lattices $\Z^d$ and $2\pi\Z^d$, respectively.

\begin{defn}\label{defn:H sk}
For any $k\in\C^d$, the subspace $H^s_k = H^s_k(W)\subset H^s(W)$ consists of the
restrictions to $W$ of functions $f\in H^s_{\loc}(\R^d)$ that satisfy the \textbf{Floquet-Bloch condition} that for any $\gamma\in\Z^d$
\begin{equation}\label{E:floquet_cond}
f(x+\gamma) = e^{ik\cdot\gamma} f(x)\ \text{a.e.}
\end{equation}
Here $H^s$ denotes the standard Sobolev space of order $s$.
\end{defn}
\noindent Note that $H^0_k = L^2(W)$.

We adopt the name \textbf{quasimomenta}\footnote{The name comes from the solid state physics \cite{AM,ReedSimon}.} for vectors $k$ in (\ref{E:floquet_cond}).

The periodicity of the coefficients of the differential operator $L(x,D)$ preserves condition (\ref{E:floquet_cond}). It thus defines an operator $L(k)$ in $L^2(W)$ with domain
$H^2_k(W)$. We could alternatively have defined $L(k)$ as the operator $L(x,D+k)$ in $L^2(\T)$ with  domain $H^2(\T)$. In the first model, $L(k)$ is realized as the $k$-independent differential expression $L(x,D)$ acting on functions in $W$ with boundary conditions depending on $k$ (which can be identified with sections of a linear bundle over the torus $\T$), while in the second one, the $k$-dependent differential expression $L(x,D+k)$ acts on the $k$-independent domain of periodic functions on $W$.

Notice that condition (\ref{E:floquet_cond}) does not change when $k$ is modified by adding an element of $2\pi\Z^d$. This, when dealing with real values of quasimomentum $k$, considering vectors $k\in B$ is sufficient.

The hypotheses on $L$ force the spectrum $\sigma(L(k))$  of the operators $L(k)$, $k\in\R^d$, to be discrete.
The following classical result (see \cite{AM,Kuch_book,ReedSimon} and the references therein) describes the spectrum of the operator $L$ in $L^2(\R^d)$:
\begin{theorem}\label{T:fl_spectrum} The following equality holds between the spectra of $L$ and $L(k)$:
\begin{equation}\label{E:fl_spectrum}
   \sigma(L)=\bigcup\limits_{k\in B}\sigma(L(k)).
\end{equation}
\end{theorem}

In other words, the spectrum of $L$ is the range of the multiple valued function
\begin{equation}\label{E:sp_function}
   k\to \lam(k): = \sigma(L(k)), k\in B.
\end{equation}
The graph of this function in $B\times \R$ is called the \textbf{dispersion relation} or \textbf{dispersion curve}\footnote{Another name of this graph is {\bf Bloch variety} of the operator $L$, e.g. \cite{Kuch_book}.} for the operator $L$:
\begin{equation}\label{E:disp_relation}
\{(k,\lam) : \lam\in\sigma\big(L(k)\big)\}.
\end{equation}
Taking into account that the operator $L(k)$ is bounded from below, one can label its eigenvalues in non-decreasing order:
\begin{equation}\label{E:eigenv}
   \lam_1(k)\leq \lam_2(k)\leq \dots .
\end{equation}
One thus gets the single-valued, continuous, and piecewise analytic \textbf{band functions} $\lam_j(k)$. The ranges of these functions constitute exactly the bands of the spectrum of $L$ shown in (\ref{E:spectrum}).

In this text, we study the Green's function at an edge of the spectrum. Without loss of generality, by shifting the operator by a constant and changing its sign if necessary, we can assume that the edge of interest is $0$ and the adjacent spectral band is $[0,a]$ for some $a>0$ and there is no spectrum for small values of $\lam$ below zero. Thus, there is a spectral gap below zero.
This means in particular that
zero is the minimal value of at least one of the band functions $\lam_j(k)$.

If zero is the bottom of the entire spectrum, then in many cases (e.g., for periodic non-magnetic Schr\"{o}dinger operators), it is known that the minimum is attained by a single band $\lambda_j(k)$
function at a single value of the quasimomentum $k_0$ (modulo $2\pi\Z^2$-periodicity) and is non-degenerate, i.e., the Hessian of $\lambda(k)$ is non-degenerate at $k_0$
\cite{KirschSimon,Pinsky,BirmSu2,BirmSu7,BirmSu8}\footnote{In the non-magnetic Schr\"{o}dinger case, $k_0=0$ \cite{KirschSimon}.}.
Although this is not necessarily always true for other spectral edges, it is expected and commonly assumed in the mathematics and physics literature that it is ``generically'' true, i.e.,
it can be achieved by small perturbation of coefficients of the operator. We thus impose the following assumption, under which we will be able to establish our main result:

{\bf Assumption A}\\
{\em There exist $k_0\in B$ and a band function $\lam_j(k)$ such that:
\begin{description}
\item[A1]  $\lam_j(k_0)=0$,
\item[A2] $|\lam_i(k)|>\delta>0$ for $i\neq j$,
\item[A3] $k_0$ is the only (up to $2\pi\Z^2$-periodicity) minimum of $\lam_j$,
\item[A4] $\lam_j(k)$ is a Morse function near $k_0$, that is, its Hessian
$H$ at $k_0$ is a positively definite quadratic form.
\end{description}
Thus, in particular,
\begin{equation}\label{E:lambda}
\lam_j(k) = \frac 12 (k-k_0)^T H(k-k_0) + O(|k-k_0|^3)
\end{equation}
is the Taylor expansion of $\lam_j$ at $k_0$.}

The parts A1 and A2 of this assumption are known to be generically true for Schr\"{o}dinger operators \cite{KloppRalston}. It is commonly believed, although not proven, that conditions A3 and A4 also are generically satisfied. Our result can be easily reformulated for a somewhat weaker condition than A3. Namely, one can assume
{\em
\begin{description}
\item[A3$\boldmath^{\prime}$]
There are finitely many (up to $2\pi\Z^2$-periodicity) minima of $\lam_j$.
\end{description}}

Validity of the non-degeneracy condition A4 is often assumed, for instance in order to define the  effective masses \cite{AM} or when studying emergence of impurity states under localized perturbations of the periodic medium \cite{BirDis1,BirDis2,BirDis3}.

In what follows, the branch $\lambda_j(k)$ plays a special role, and from now on we  use the shorthand notation
\begin{equation}\label{E;short}
\lambda(k):=\lambda_j(k).
\end{equation}

Let us also introduce the notation $\vp(k_0,x)$ for a $\Z^d$-periodic function of $x$ such that $\psi(x)=e^{ik_0\cdot x}\vp(k_0,x)$ is the (unique up to a constant factor) eigenfunction of $L(k_0)$ with the eigenvalue $\lam(k_0)$:
\begin{equation}
 L(k_0)\psi=\lam(k_0)\psi.
\end{equation}

We are ready now to state the main result of this text:

\begin{theorem}\label{T:main} Let $d\geq 3$, the operator $L$ satisfy the assumption A, and $R_{-\epsilon}=(L+\epsilon)^{-1}$ for a small $\epsilon>0$ denote the resolvent of $L$ near the spectral edge $\lambda=0$ (which exists, due to Assumption A). Then:
\begin{enumerate}
\item For any $\phi,\psi\in L^2_{comp}(\R^d)$,
$$
R_{-\epsilon}\phi\cdot\psi\mathop{\to}\limits_{\epsilon\to 0} R\phi\cdot\psi
$$
for an operator $R:L^2_{comp}(\R^d)\mapsto L^2_{loc}(\R^d)$.
\item The Schwartz kernel $G(x,y)$ of $R$, which we will simply call {\bf the Green's function of $L$}, has the following asymptotics when
$|x-y|\to\infty$:
\begin{align}
\label{E:expansion}
 &G(x,y) \nonumber \\
 &= \frac{\pi^{-d/2}\Gamma(\frac{d-2}{2}) e^{i(x-y)\cdot k_0}}
{2\sqrt{\det H}\, |H^{-1/2}(x-y)|^{d-2}}
\frac{\vp(k_0,x) \overline{\vp(k_0,y)}} {\|\vp(k_0)\|_{L^2(\T)}^2} \left(1 +  O\left(|x-y|^{-1}\right) \right) \\&+r(x,y),\nonumber
\end{align}
where $r(x,y)=O(|x-y|^{-d+1})$ .
\end{enumerate}
\end{theorem}
Here $H$ is the Hessian matrix from (\ref{E:lambda}) and the standard notation $L^2_{comp}(\R^d)$ and $L^2_{loc}(\R^d)$ are used for the spaces of square integrable functions with compact support and locally square integrable functions on $\R^d$ correspondingly.

The idea of the proof is to show that only one branch of the dispersion relation (namely, the one appearing in Assumption A) governs the asymptotics
with the crucial term being the Hessian $H$ at the point $k_0$. This appears like an analog of homogenization (which is an effect that occurs at the bottom of the spectrum only), and indeed it is a version of it, along with other effects described in \cite{BS3,Birman-homog,BirmSu8,KuchPin1,Kuch_Molch,KuchPin2}.

We prove Theorem \ref{T:main} in Section \ref{S:proof} after introducing the tools of Floquet theory \cite{ReedSimon,Kuch_book} in Section \ref{S:floquet}.
Section \ref{S:remarks} contains  final remarks and is followed by an Acknowledgments Section.

\section{A Floquet reduction of the problem}\label{S:floquet}

In this section we use Assumption A and reduce  the problem of obtaining asymptotics for the Green's function $G(x,y)$ to the problem of finding asymptotics of a scalar integral expression.
The scalar integral we obtain resembles the Green's function of the Laplacian.

\subsection{Floquet transform}\label{SS:floquet}

Let $\gamma\in\R^d$ and a function $f(x)$ be defined on $\R^d$. We then denote by $f_\gamma$ the $\gamma$-shifted version of $f$: $f_\gamma(x)= f(x+\gamma)$. Assuming that $f(x)$ decays sufficiently fast (e.g., $f$ is compactly supported, although $f\in L^2$ would suffice, since it guarantees the $L^2_{loc}$-convergence in (\ref{eqn:hat f(k)}) below), the following transform
plays the role of Fourier transform when studying periodic problems (in fact, it \textbf{is} the Fourier transform with respect to the lattice $\Z^d$ of periods):
\begin{defn}\label{defn:Floquet transform}
The \textbf{Floquet transform} (also sometimes called \textbf{Gelfand transform})
\begin{equation}\label{E:transform}
f(x)\mapsto \hat{f}(k,x)
\end{equation}
maps a function $f(x)$ of $x\in\R^d$ into a function $\hat f(k,x)$ of $(k,x)\in B\times W$, defined as follows:
\begin{equation}\label{eqn:hat f(k)}
\hat f(k) = \hat f(k,x) = \sum_{\gamma\in\Z^d}f_\gamma(x) e^{-ik\cdot\gamma}
= \sum_{\gamma\in\Z^d}f(x+\gamma)e^{-ik\cdot\gamma}.
\end{equation}
\end{defn}
Notice that $\hat{f}(k,x)$ is $2\pi\Z^d$-periodic with respect to $k$ (recall that $\G^*:=2\pi\Z^d$ is the dual lattice to $\G:=\Z^d$). Thus, it can be naturally interpreted\footnote{Formally speaking, this interpretation requires to change the quasimomentum variable $k$ to the {\bf Floquet multiplier} variable $z=(e^{ik_1},... ,e^{ik_d})$. We will abuse notation (e.g., in Lemma \ref{L:floquet}), staying with the same variable $k$.} as a function on the torus $\T^*$ with respect to $k$. With respect to the $x$-variable, though, the function is not periodic, but rather satisfies the {\bf cyclic} (or {\bf Floquet}) condition
\begin{equation}\label{E:cyclic}
\hat f(k,x+\g)  =e^{ik\cdot\g}\hat f(k,x)
\end{equation}
for all $\g\in\Z^d$. Thus, the function $\hat f(k,x)$ is completely determined by its values on $B\times W$.
We  usually consider it as a function $\hat f(k,\cdot)$ of $k\in B$ with values in a space of functions defined on $W$.

The next result collects some well known properties of the Floquet transform (e.g., \cite{Kuch_book}).
\begin{lemma}\label{L:floquet}
\indent
\begin{enumerate}
\item The Floquet transform is an isometry (up to a scalar factor) between $L^2(\R^d)$
and $\int\limits_B^{\oplus} L^2(W)\, dk = L^2(B, L^2(W))$.

\item Its inversion for any $f \in L^2(\RR^n)$ is given by the following two equivalent formulas:
\begin{equation}\label{E:Gelf_inversion}
f(x)=\int\limits_{\T^*} \fh (k,x) \dn k, \,\, x \in \R^d
\end{equation}
and
\begin{equation}\label{E:Gelf_inversion2}
f(x)=\int\limits_{\T^*} \fh (k,x-\g) e^{ik \cdot \g} \dn k, \,\,
x \in W+\g.
\end{equation}
 \item Under the Floquet transform, the action of the operator $L$ in $L^2(\R^d)$ is transformed into $$\int\limits^\bigoplus_B L(k),$$ i.e.,
\[
\widehat{Lf}(k)= L(k)\widehat{f}(k)
\]
for any $f\in H^2(\R^d)$.
\item Let $\phi (k,x)$ be a function on $\R_k^d\times\R^d_x$ such that it is $2\pi\Z^d$-periodic with respect to $k$, belongs  to the space $H^s_k$ for each $k$
(see Definition \ref{defn:H sk}), and is an infinitely differentiable function of $k$ with values in $H^s_k$. Then $\phi$ is the Floquet transform of a function $f\in H^s(\R^d)$.
Moreover, for any compact set $K\subset\R^d$ the norm $\|f\|_{H^s(K+\g)}$ decays faster than any power of $|\g|$. In particular, if $s>d/2$, then
    $$
    |f(x)|\leq C_N(1+|x|)^{-N}\mbox{ for any }N.
    $$
\end{enumerate}
\end{lemma}
Here $\dn k$ denotes the normalized Haar measure on $\T^*$, or the corresponding Lebesgue measure on $B$.

\subsection{The Floquet reduction}\label{SS:fl_reduction}

We start by considering the resolvent in the spectral gap:
\begin{equation*}
R_{-\epsilon}f=(L+\epsilon)^{-1}f \mbox{ for } f\in L^2_{comp}(\R^d) \mbox{ and a small positive }\epsilon.
\end{equation*}
Applying the Floquet transform,
\begin{equation*}
\widehat{R_{-\epsilon}f}(k)=(L(k)+\epsilon)^{-1} \hat f(k).
\end{equation*}
Given another compactly supported function $\psi\in L^2_{comp}(\R^d)$, we begin our analysis with the sesquilinear form
\begin{equation}\label{E:bilinform}
    (R_{-\epsilon}f,\psi)=\int_{T^*}\left((L(k)+\epsilon)^{-1} \hat f(k),\hat \psi(k)\right)\dn k.
\end{equation}
In the following lemma, we prove the first claim of Theorem \ref{T:main} and introduce an expression
for the operator $R=\lim_{\epsilon\to 0}R_{-\epsilon}$ that is the basis of our study of the asymptotics of the Green's function.
\begin{lemma}\label{L:bilin} Let the Assumption A be satisfied. Then the following equality holds:
\begin{equation}\label{E:limit}
    \lim_{\epsilon\to 0}(R_{-\epsilon}f,\psi)=\int_{T^*}\left(L(k)^{-1} \hat f(k),\hat \psi(k)\right)\dn k.
\end{equation}
Thus, the Green's function is the kernel of the operator $R$ defined as follows:
\begin{equation}\label{E:R}
    \widehat{Rf}(k)=L(k)^{-1} \hat f(k).
\end{equation}
The integral in the right hand side of (\ref{E:limit}) is absolutely convergent for $f,\psi\in L^2_{comp}$.
\end{lemma}

\begin{proof} The \textbf{complex Bloch variety} of the operator $L$, denoted by $\Sigma$,  consists of all complex pairs $(k,\lambda)\in\C^{d+1}$ such that the operator $L(k)-\lambda$ on the torus $\T$ is not invertible. This variety is an analytic subset of co-dimension $1$ (e.g., \cite{Kuch_book}). Moreover,
there exist an entire scalar function $h(k,\lambda)$ on $\C^{d+1}$ and an entire operator-valued function $I(k,\lambda)$, such that

1) $h$ vanishes only on (the nonsingular) $\Sigma$,
has simple zeros at all points of $\Sigma$, and

2) outside $\Sigma$, the inverse operator $(L(k)-\lambda)^{-1}=h(k,\lambda)^{-1}I(k,\lambda)$ \cite{Kuch_book,Wilcox}.

In particular, for a small $\epsilon>0$, our operator of interest $(L(k)+\epsilon)^{-1}=h(k,-\epsilon)^{-1}I(k,-\epsilon)$. Due to  assumption A, for small $\epsilon>0$ and quasimomenta $k$ with small imaginary parts, the function $h(k,-\epsilon)^{-1}$ is equal, up to a non-vanishing analytic factor, to $(\lambda(k)+\epsilon)^{-1}$. Thus, for such values of $k$ and $\epsilon$, one can write the bilinear form $(R_{-\epsilon}f,\phi)$ as
\begin{equation}\label{E:simplified}
    (R_{-\epsilon}f,\phi)=\int\limits_{\T^*}\frac{\left(M(k,\epsilon)\hat f(k),\hat\phi (k)\right)}{\lambda(k)+\epsilon}\dn k.
\end{equation}
Here $M(k,\epsilon)$ is an operator-function in $L^2(\T)$, analytic for small $\epsilon$ and nearly real $k$. The function $\lambda(k)$, due to Assumption A, is a non-negative function with a single simple non-degenerate zero on $\T$. Also, due to the compactness of support of $f(x)$ and $\psi(x)$ and (\ref{eqn:hat f(k)}), the functions $\hat f(k), \hat \psi(k)$ are analytic with respect to $k$. Thus, for $\epsilon>0$, the integrand in (\ref{E:simplified}) is nonsingular. Since $d\geq 3$, it is still integrable when $\epsilon=0$. Hence, a straightforward application of the dominated convergence theorem finishes the proof.
\end{proof}
\begin{remark}\label{R:limit}
   Lemma \ref{L:bilin} could be proven without using the results of \cite{Kuch_book,Wilcox} concerning the representation of the resolvent $(L(k)-\lambda)^{-1}=h(k,\lambda)^{-1}I(k,\lambda)$ (which requires some regularized infinite determinant techniques). One could instead achieve the same using eigenfunction expansions that are employed below.
\end{remark}

\subsection{Singling out the principal term in $R$}\label{SS:principal}
As we have just observed, the Green's function in question is the kernel of the operator $R$ that acts according to (\ref{E:R}): for a compactly supported $L^2$-function $f$, the Floquet transform of $Rf$ is given by  $\widehat{Rf}(k)=L(k)^{-1} \hat f(k)$. In other words,
\begin{equation}\label{E:R_expl}
   Rf(x)=\int\limits_{\T^*}L(k)^{-1} \hat f(k,x)\dn k \mbox{ for }x\in\R^d,
\end{equation}
or
\begin{equation}\label{E:R_expl2}
   Rf(x+\g)=\int\limits_{\T^*}e^{ik\cdot\g}L(k)^{-1} \hat f(k,x)\dn k \mbox{ for }x\in W, \g\in\Z^d.
\end{equation}
In this section, we will single out the part of the integral in (\ref{E:R_expl2}) that is
responsible for the principal term of the Green's function asymptotics. The first thing to notice is that in order to determine the Schwartz kernel of $R$, it is sufficient to apply the operator to smooth functions. Thus, we may assume that  $f\in C^\infty_0(\R^d)$.

Let us localize the integral around the minimum point $k_0$ of $\lambda(k)$. In order to do so, we first recall \cite{Kuch_book,Wilcox} that
in a neighborhood of $k_0$, there exist  functions $\psi(k,x)$ and $\vp(k,x)$ that are analytic in $k$ and satisfy
\begin{enumerate}
\item $\psi(k,x) = e^{ik\cdot x}\vp(k,x)$,
\item $\vp(k,x)$ is $\Z^d$-periodic in $x$, and
\item $L(k)\psi(k) = \lam(k)\psi(k)$.
\end{enumerate}
There exists an analytic spectral projector $P(k)$ that projects $L^2(\T)$ onto the span of $\psi(k)$.

Let $\nu,\eta \in C^\infty_0(B)$ (here $B=[-\pi,\pi]^d$) be cut-off functions so that $\nu\equiv 1$ near $k_0$, $\eta|_{\supp \nu}=1$, and $\psi(k)$ is well-defined on $\supp\eta$.
We can decompose
\[
\hat f = \nu\hat f + (1-\nu)\hat f.
\]
The operator $L(k)$ is invertible (in $x$) whenever $k$ stays away from $k_0$. Thus,
\[
\hat u_g(k,x) = L^{-1}(k)(1-\nu(k))\hat f(k,x)
\]
is well-defined and smooth with respect to $(k,x)$ (recall that the operator $L(k)$ is elliptic and depends analytically on $k$). The smoothness of
$\hat u_g$ means (according to the statement 4 of Lemma \ref{L:floquet}) that $u_g$ has rapid decay in $x$ (we thus chose the subscript $g$ in $u_g$ to indicate that this is a ``good'' part of $u$).

Next, we must solve
\[
L(k)\hat u(k) = \nu(k)\hat f(k).
\]
We will look for a solution $\hat u$ of the form $\hat u = \hat u_1 + \hat u_2$ where $\hat u_1 = P(k)\hat u_1$,
$\hat u_2 = Q(k)\hat u_2$, and $Q(k) = I-P(k)$. In this spirit, let $\hat f_1 = P(k)\nu(k)\hat f$, and $\hat f_2 = Q(k)\nu(k)\hat f$. Observe that
\[
\hat f_1 = P(k)\eta(k)\nu(k)\hat f = \eta(k) P(k)\nu(k)\hat f = \eta(k)\hat f_1
\]
since $\eta$ is a function of $k$ and not of $x$.
With the two orthogonal projection operators $P(k)$ and $Q(k)$, the problem of solving $L(k) \hat u(k) = \hat f(k)$ decomposes into
\[
\left(\begin{array}{c|c}
L(k)P(k) & 0 \\ \hline
0 & L(k)Q(k)
\end{array}\right)
\begin{pmatrix} \hat u_1 \\ \hat u_2 \end{pmatrix} = \begin{pmatrix} \hat f_1 \\ \hat f_2 \end{pmatrix}.
\]
In this block-matrix structure, the first column corresponds to a one-dimensional subspace (the range of $P(k)$), while the second one corresponds to the infinite dimensional range of $Q(k)$. Since the operator $L(k)Q(k)$, considered on the range of $Q(k)$, has no spectrum near zero, the operator
$(L(k)Q(k))^{-1}$ is well-defined and so is $\hat u_2$. Additionally, $\hat u_2$ is smooth, so according to Lemma \ref{L:floquet} again, $u_2$ decays rapidly
as $|x|\to\infty$.

We know that $\hat u_1 = P(k) \hat u_1$ and $L(k)\psi(k) = \lam(k)\psi(k)$. Putting these facts together, we need $\hat u_1$ to
satisfy the equality
$$
L(k)P(k) \hat u_1 = P(k) \hat f_1 = \eta(k)P(k)\hat f_1,
$$
which can be rewritten consecutively in the following forms:
\begin{align*}
L(k) \frac{\big( \hat u, \psi(k) \big)_{\T}}{\|\psi(k)\|_{L^2(\T)}^2} \psi(k)
&= \frac{\eta(k)\big(\hat f_1,\psi(k) \big)_{\T}}{\|\psi(k)\|_{L^2(\T)}^2} \psi(k), \\
\frac{\big( \hat u_1, \psi(k)\big)_{\T}}{ \|\psi(k)\|_{L^2(\T)}^2 } L(k) \psi(k)
&= \frac{\eta(k)\big(\hat f_1,\psi(k) \big)_{\T}}{\|\psi(k)\|_{L^2(\T)}^2}  \psi(k), \\
\frac{\big( \hat u_1, \psi(k)\big)_{\T}}{\|\psi(k)\|_{L^2(\T)}^2} \lam(k) \psi(k)
&= \frac{\eta(k) \big(\hat f_1,\psi(k) \big)_{\T}}{\|\psi(k)\|_{L^2(\T)}^2} \psi(k). \\
\end{align*}
Therefore, it is useful to define
\[
\hat u_1(k,x) := \frac{ \eta(k) (\hat f_1(k,\cdot),\psi(k,\cdot))_{\T} \psi(k,x)}
{\|\psi(k)\|_{L^2(\T)}^2 \lam(k)}.
\]
Applying the inverse Floquet transform, we have for $n\in\Z^d$,
\[
u_1(x+n) = \frac{1}{(2\pi)^d} \int_{B} e^{i n\cdot k} \frac{ \eta(k)(\hat f_1(k,\cdot),\psi(k,\cdot))_{\T} \psi(k,x)}
{\|\psi(k)\|_{L^2(\T)}^2 \lam(k)} \, dk.
\]

\subsection{From $u_1$ to a ``reduced'' Green's function}

We now pass from $u_1(x+n)$ to computing, up to lower order terms at infinity, the Green's function. Below, we will define a function $\Go(x,y)$ that provides
the leading term of the asymptotics of the Green's function
$G(x,y)$. To do so, our plan is to isolate a scalar integral, whose asymptotic behavior is responsible for the asymptotics of the Green's function.

We introduce $\Go(x,y)$ as follows:
\[
u_1(x+n) = \int_{\R^d} G_0(x+n,y) f_1(y)\, dy.
\]
Observe that
\begin{align*}
u_1(x+n)
&= \frac{1}{(2\pi)^d} \int_B \int_\T  e^{in\cdot k} \hat f_1(k,y) \eta(k) \frac{\overline{\psi(k,y)}\psi(k,x)}
{\|\psi(k,\cdot)\|_{\T}^2 \lam(k)}\, dy\, dk \\
&= \frac{1}{(2\pi)^d} \int_B \eta(k) \int_\T \sum_{\gamma\in\Z^d} e^{in\cdot k} f_1(y+\gamma) e^{-ik\cdot\gamma}
\frac{\overline{\psi(k,y)}\psi(k,x)} {\|\psi(k,\cdot)\|_{\T}^2 \lam(k)}\, dy\, dk \\
&= \frac{1}{(2\pi)^d} \int_B \eta(k)\sum_{\gamma\in\Z^d}\int_{\gamma+\T}  e^{in\cdot k} f_1(y) e^{-ik\cdot\gamma}
\frac{\overline{\psi(k,y-\gamma)}\psi(k,x)} {\|\psi(k,\cdot)\|_{\T}^2 \lam(k)}\, dy\, dk
\end{align*}
We notice that
\[
\psi(k,y-\gamma) = e^{ik\cdot(y-\gamma)}\vp(k,y) = e^{-ik\cdot\gamma}\psi(k,y),
\]
so $\overline{\psi(k,y-\gamma)} = e^{ik\cdot\gamma} \overline{\psi(k,y)}$ and
\[
u_1(x+n) = \int_{\R^d} f_1(y) \left( \frac{1}{(2\pi)^d} \int_B   e^{i n\cdot k} \eta(k)
\frac{\overline{\psi(k,y)}\psi(k,x)} {\|\psi(k,\cdot)\|_{\T}^2 \lam(k)}\, dk \right) \, dy.
\]
Thus, we have
\[
\Go(x+n ,y) = \frac{1}{(2\pi)^d} \int_B   e^{in\cdot k} \eta(k)
\frac{\overline{\psi(k,y)}\psi(k,x)} {\|\psi(k,\cdot)\|_{\T}^2 \lam(k)}\, dk
\]
Note that if $y \in\T$ and $\gamma\in\Z^d$, then
\[
\Go(x+n ,y+\gamma ) = \frac{1}{(2\pi)^d} \int_B   e^{i(x+n-(y+\gamma))\cdot k} \eta(k)
\frac{\overline{\vp(k,y)}\vp(k,x)} {\|\vp(k,\cdot)\|_{\T}^2 \lam(k)}\, dk
\]
or more simply,
\begin{equation}\label{eqn:formula for G(x,y)}
\Go(x,y) = \frac{1}{(2\pi)^d} \int_B   e^{i(x-y)\cdot k} \eta(k)
\frac{\overline{\vp(k,y)}\vp(k,x)} {\|\vp(k,\cdot)\|_{\T}^2 \lam(k)}\, dk.
\end{equation}

\section{Asymptotics of the Green's function}\label{S:proof}
We start with deducing asymptotics in the reduced case.
\subsection{Reduced Green's function $\Go$}
Let
\[
\rho(k) := \frac{\overline{\vp(k,y)}\vp(k,x)} {\|\vp(k,\cdot)\|_{\T}^2}
\]
Then $\rho$ is a smooth function in $k$ near $k_0$, $x$, and $y$.
Let $\mu_0(k-k_0)=\eta(k)$, so $\mu_0$ is a cutoff function supported near $0$. Then
\begin{equation}\label{E:mu0}
\Go(x,y) = \frac{1}{(2\pi)^d} \int_{B} e^{i(x-y)\cdot k} \frac{\mu_0(k-k_0)\rho(k)}{\lam(k)} \, dk.
\end{equation}
It will be shown in the next subsection that the integral in (\ref{E:mu0}) provides the main term of the asymptotics for the full Green's function
$G(x,y)$.

By Taylor expanding around $k_0$, we can write
$$
\rho(k) = \rho(k_0) + (k-k_0)\cdot \rho_1(k).$$
The integral then becomes
\begin{align}\label{eqn:u1(x+n)}
\Go(x,y) &=  \frac{\rho(k_0)}{(2\pi)^d} \int_{B} e^{i(x-y)\cdot k} \frac{ \mu_0(k-k_0)}{\lam(k)} \, dk \nonumber \\
&+  \frac{1}{(2\pi)^d} \int_{B} e^{i(x-y)\cdot k} \frac{\mu_0(k-k_0)(k-k_0)\cdot \rho_1(k)}{\lam(k)}\, dk.
\end{align}
We will see that the main term comes from the first integral on the righthand side of
\eqref{eqn:u1(x+n)}. It is this integral that we address first.

Assumption A2 means that the matrix $H(\lam,k_0)$ from Assumption A4 is positive definite.
Let $A = H(\lam,k_0)$. Then
\begin{equation}\label{E:lambdaA}
\lam(k) =  \frac 12 (k-k_0)^T A (k-k_0) + g_1(k-k_0),
\end{equation}
where $g_1=O(|k-k_0|^3)$.

Let $A^{1/2}$ be the positive definite square root of $A$. Then, after the changes of variables $z=k-k_0$ and $\xi = A^{1/2}z$, we get
\begin{align*}
& \int_{B} e^{i(x-y)\cdot k} \frac{ \mu_0(k-k_0)}{\lam(k)} \, dk\\ &=
e^{i(x-y)\cdot k_0} \int_{\R^d} e^{i (A^{-1/2}(x-y))^T  A^{1/2}z} \frac{ \mu_0(A^{-1/2} A^{1/2}z)}{\frac 12 (A^{1/2}z)^T  A^{1/2}z + g_1(A^{-1/2} A^{1/2}z)} \, dz \\
&= \frac{2 e^{i(x-y)\cdot k_0}} {(\det A)^{1/2}} \int_{\R^d} \frac { e^{i A^{-1/2}(x-y) \cdot \xi}} {\xi\cdot\xi + g(\xi)} \mu(\xi)\, d\xi.
\end{align*}
Here $g= 2 g_1\circ A^{-1/2} = O(|\xi|^3)$ and $\mu = \mu_0\circ A^{-1/2}$.

Let $x_0 = A^{-1/2}(x-y)$. We have now reduced the problem to finding the asymptotics of the following integral in $\R^n$:
\begin{equation}\label{eqn:rotate and dilate, lam=0}
\frac{1}{(2\pi)^d} \int_{\R^d} \frac { e^{ix_0\cdot\xi} }{\xi\cdot\xi + g(\xi)} \mu(\xi)\, d\xi,
\end{equation}
where $g(\xi) = O(|\xi|^3)$. The main goal is to eliminate the extra term $g(\xi)$, after which we have a situation similar to the ordinary Laplacian.
To do this, we prove the following proposition:

\begin{prop}\label{thm:fourier transform arg for lam=0}
Let $A$, $g$, and $\mu$ be as in (\ref{eqn:rotate and dilate, lam=0}). If $|x_0|\geq 2$, then
\[
\frac{1}{(2\pi)^d}\int_{\R^d}
\frac{e^{i x_0\cdot\xi}\mu(\xi)} {\xi\cdot\xi + g(\xi)} \, d\xi
= \frac{C_{d}}{|x_0|^{d-2}}\big( 1+ O(|x_0|^{-1}) \big), |x_0|\to \infty,
\]
where
\begin{equation}\label{eqn:Cd}
C_d = \frac{\Gamma(\frac d2 -1)}{4 \pi^{d/2}}.
\end{equation}
\end{prop}

In the following proof we use the notation
$c_z$ for a constant that depends on $z$ only. The value of $c_z$ may change from line to line.

\begin{proof}If the support of $\mu$ is small enough, then since $g(\xi) = O(|\xi|^3)$, we have
\begin{align*}
\frac{\mu(\xi)}{|\xi|^2+g(\xi)} &= \frac{\mu(\xi)}{|\xi|^2} \frac{1}{1 + |\xi|^{-2}g(\xi)} \\
&= \frac{\mu(\xi)}{|\xi|^2} + \frac{\mu(\xi)}{|\xi|^2} \sum_{n=1}^{d+1} (-1)^n \frac{g(\xi)^n}{|\xi|^{2n}} + r_1(\xi)
\end{align*}
where $r_1(\xi) \in C^d_0(\R^d)$. Next, by Taylor expanding $g(\xi)$, we can write
\begin{equation}\label{E:mu_expansion}
\frac{\mu(\xi)}{|\xi|^2+g(\xi)} = \frac{\mu(\xi)}{|\xi|^2}  + \mu(\xi) \sum_{n=1}^{d+1} \sum_{3n \leq |\alpha| \leq 2n+2+d}\frac{g_{\alpha,n}\xi^\alpha}{|\xi|^{2n+2}} + r(\xi)
\end{equation}
for some constants $g_{\alpha,n}$ and $r\in C^d_0(\R^d)$.

We start by handling the lower order terms in (\ref{E:mu_expansion}) first, namely the terms with $\xi^\alpha/|\xi|^{2n+2}$ and $r(\xi)$.
\subsubsection{Terms of the form $\xi^\alpha/|\xi|^{2n+2}$}
\begin{lemma}
Suppose that $h(x)$ is such that its Fourier transform $\tilde h(\xi)$ is a function homogeneous of order $-1$. Then
\begin{equation}\label{eqn:F.t. of homog funct deg -1}
\Big| \int_{\R^d} e^{ix_0\cdot \xi} \tilde h(\xi) \mu(\xi)\, d\xi\Big| \leq \frac{C}{|x_0|^{d-1}}
\end{equation}
for $|x_0|\geq 2$.
\end{lemma}
\begin{proof}
Let us observe that
\begin{align*}
\int_{\R^d} e^{ix_0\cdot \xi} \tilde h(\xi) \mu(\xi)\, d\xi &= \int_{\R^d} h(x) \widetilde{e^{ix_0\cdot\xi}\mu}(x)\, dx\\
&= \int_{\R^d} h(x) \tilde \mu(x-x_0)\, dx
= \int_{\R^d} h(x-x_0) \tilde\mu(x)\, dx.
\end{align*}
Since $\mu\in C^\infty_0(\R^d)$, its inverse Fourier transform $\tilde\mu$ is a Schwartz function and hence rapidly decreasing, i.e., there exist constants $C_\ell$ for $\ell\geq 0$
so that for
\[
|D^\ell \tilde\mu(x)| \leq \frac{C_\ell}{|x|^\ell}.
\]

Next, for $|x_0|\geq 2$, it follows that
\begin{align*}
\Big|  \int_{\R^d} h(x-x_0) \tilde\mu(x)\, dx \Big|
&\leq \Big| \int_{|x-x_0|\leq \frac 12|x_0|} \frac{1}{|x-x_0|^{d-1}} \frac{C_d}{|x|^d}\, dx\Big|
+ \Big| \int_{|x|<1} \frac{2^{d-1}}{|x_0|^{d-1}} C_0\, dx\Big| \\
&+ \Big| \int_{\atopp{|x-x_0|\geq \frac{|x_0|}2}{|x|\geq 1}} \frac{2^{d-1}}{|x_0|^{d-1}} \frac{C_{d+1}}{|x|^{d+1}}\, dx \Big|
\leq \frac{c_{d,\mu}}{|x_0|^{d-1}}
\end{align*}
for some constant $c_{d,\mu}$ that depends on $\mu$ and $d$. This proves the claim of the lemma.
\end{proof}

We will now use an integration by parts argument to reduce to the case of degree $-1$ homogeneity. If $\alpha = (\alpha_1,\dots,\alpha_n)$, then
\[
\frac{\p}{\p\xi_j} \Big(\frac{\xi^\alpha}{|\xi|^\ell}\Big) = \frac{-\ell \xi_j \xi^\alpha}{|\xi|^{\ell+2}} + \frac{\alpha_j \xi_1^{\alpha_1}\cdots \xi_j^{\alpha_j-1}\cdots \xi_n^{\alpha_n}}{|\xi|^\ell},
\]
a function that is homogeneous of degree $|\alpha|-1-\ell$. Assume that $x_0 = (x_0^1,\dots,x_0^n)$ and $|x_0| \sim |x_0^1|$.
Then for $|\alpha|\geq \ell-1$ and $q:=|\alpha|- \ell+1$ one has
\begin{align*}
\Big|\int_{\R^d} e^{ix_0\cdot\xi} \mu(\xi) \frac{\xi^\alpha}{|\xi|^\ell}\, d\xi\Big|
&\leq \frac{C}{|x_0^1|^{|\alpha|-\ell+1}}\sum_{j=0}^{|\alpha|-\ell+1} \Big| \int_{\R^d} \frac{\p^j}{\p\xi_1^j}\Big(\frac{\xi^\alpha}{|\xi|^\ell}\Big)
\frac{\p^{q-j}\mu(\xi)}{\p \xi_1^{q-j}} \, dx\Big|.
\end{align*}
If $j=|\alpha|-\ell+1$, then we can use (\ref{eqn:F.t. of homog funct deg -1}) with $\tilde h(\xi) = \frac{\p^{|\alpha|-\ell+1}}{\p\xi_1^{|\alpha|-\ell+1}} \big(\frac{\xi^\alpha}{|\xi|^\ell}\big)$ and compute
that
\[
\frac{1}{|x_0^1|^{|\alpha|-\ell+1}}\Big| \int_{\R^d} \frac{\p^{|\alpha|-\ell+1}}{\p\xi_1^{|\alpha|-\ell+1}}\Big(\frac{\xi^\alpha}{|\xi|^\ell}\Big) \mu(\xi)\, d\xi\Big|
\leq C \frac{1}{|x_0|^{|\alpha|-\ell+1+d-1}}.
\]

If $j\neq |\alpha|-\ell+1$, then there exists $c$ so that $c^{-1} \leq |\xi|\leq c$ if $\xi\in\supp \frac{\p^j \mu}{\p\xi_1^j}$
since $\mu\in C^\infty_0$ and $\mu\equiv 1$ in a neighborhood of $0$. In this case, we can integrate by parts as needed to obtain decay of $O(|x_0|^{-d+1})$ (or of any higher power of
$|x_0|$).

Combining our estimates, there exists a constant $C = C(\mu,d,g)$ so that
\[
\bigg| \int_{\R^d} e^{ix_0\cdot \xi}  \mu(x) \sum_{n=1}^{d+1} \sum_{3n \leq |\alpha| \leq 2n+2+d}\frac{g_{\alpha,n}\xi^\alpha}{|\xi|^{2n+2}}\, d\xi\bigg|
\leq \frac{C}{|x_0|^{d-1}}.
\]

\subsubsection{Estimate for the $r(\xi)$ term}

Since $r(\xi)\in C^d_0(\R^d)$, we may integrate by parts $(d-1)$ times and observe that
\[
\bigg|\int_{\R^d}
e^{i |x_0|\xi_1} r(\xi)  \, d\xi\bigg| = \frac{1}{ |x_0|^{d-1}} \int_{\R^d}  \bigg|\frac{\p^{d-1} r(\xi)}{\p \xi_1^{d-1}}\bigg| d\xi \leq \frac{c_{d}} {|x_0|^{d-1}}.
\]

\subsubsection{The main term estimate}

We will show that for the constant $C_d$ from (\ref{eqn:Cd}), the following standard asymptotics holds:
\[
\frac{1}{(2\pi)^d}\int_{\R^d} \frac{e^{i x_0\cdot \xi} \mu(\xi)} {|\xi|^2} \, d\xi
= \frac{C_{d}}{|x_0|^{d-2}} + O(|x_0|^{-d+1}).
\]
Let $\theta$ be a radial, smooth approximation of the identity with $\supp\theta \subset B(0,\frac 14)$ and $\int_{\R^d} \theta\, dx =1$.
Recall that the volume of the unit ball in $\R^d$ is $\omega_d = \frac{\pi^{d/2}}{\Gamma(\frac d2 +1)}$ and the Newtonian potential is
\[
\Psi(x) = \frac{1}{d(d-2)\omega_d} \frac{1}{|x|^{d-2}}.
\]
Since $\Psi$ is a harmonic function on $B(x_0,\frac 14)$, it follows from the mean-value property that
\[
\Psi*\theta(x_0) = \Psi(x_0).
\]
Therefore, denoting by $\beta (x)$ the function whose Fourier transform is $\mu(\xi)$, we get
\[
\int_{\R^d} \frac{e^{i x_0\cdot\xi} \mu(\xi)} {|\xi|^2} \, d\xi
=\int_{\R^d} \frac{e^{i x_0\cdot\xi} \tilde\theta(\xi)} {|\xi|^2} \, d\xi + \int_{\R^d} \frac{e^{i x_0\cdot\xi}(\tilde\beta - \tilde\theta)(\xi)} {|\xi|^2} \, d\xi.
\]
The inverse Fourier transform applied to $\tilde\theta(\xi)/|\xi|^2$ yields
\[
\frac{1}{(2\pi)^d}\int_{\R^d} \frac{e^{i x_0\cdot\xi} \tilde\theta(\xi)} {|\xi|^2} \, d\xi
= \Psi*\theta(x_0) = \Psi(x_0).
\]
Also, since $\tilde\theta(0) = \tilde\beta(0)=1$, we can write
\[
\frac{(\tilde\beta - \tilde\theta)(\xi)} {|\xi|^2} = \frac{1}{|\xi|^2}\sum_{1\leq|\alpha|\leq d+1} \xi^\alpha \mu_\alpha(\xi) + \mu'(\xi)
\]
where $\mu_\alpha$ is rapidly decreasing in $|\xi|$ and $\mu' \in C^d_0(\R^d)$.

Now the argument provided earlier in the proof leads to the desired error estimate
\begin{equation}\label{E:error est with theta_r}
\int_{\R^d}e^{i x_0\cdot\xi}\Big(\frac{1}{|\xi|^2}\sum_{1\leq|\alpha|\leq d+1} \xi^\alpha \mu_\alpha(\xi) + \mu'(\xi)\Big)\, d\xi
= O(|x_0|^{1-d}).
\end{equation}
 Finally, since $\mu_\alpha$ depends on $\theta$ and $\beta$,
so the $O$-term in (\ref{E:error est with theta_r}) is bounded independently of $|x_0|$.
\end{proof}

In order to finish the proof of the asymptotics of the reduced Green's function $\Go$, it remains to bound the last term in (\ref{eqn:u1(x+n)}):
\[
 \frac{1}{(2\pi)^d} \int_{B} e^{i(x-y)\cdot k} \frac{(k-k_0)\cdot\rho_1(k)}{\lam(k)}\mu_0(k)\, dk.
\]
This integral can be estimated similarly to the technique used for controlling the error terms in Proposition \ref{thm:fourier transform arg for lam=0}.
Let us use a Taylor expansion in $(k-k_0)\cdot \rho_1(k)$ around $k_0$. In the notation from the proof of Proposition \ref{thm:fourier transform arg for lam=0}, we have
\[
\frac{(k-k_0)\cdot \rho_1(k)}{\lam(k)} \mu_0(k-k_0) = \mu(\xi) \sum_{n=0}^{d+1} \sum_{3n+1\leq |\beta|\leq 2n+2+d} c_{\beta,n} \frac{\xi^\beta}{|\xi|^{2n+2}} + R(\xi).
\]
The remainder function $R(\xi)$ belongs to $C^d_0(\R^d)$, and as in the proof of Proposition \ref{thm:fourier transform arg for lam=0} it follows that
\[
\frac{1}{(2\pi)^d} \int_{B} e^{i x_0\cdot\xi}\Big( \frac{\mu(\xi)}{|\xi|^2} \sum_{1\leq |\beta|\leq d} c_\beta \xi^\beta + R(\xi)\Big) \, d\xi = O(|x_0|^{1-d}).
\]

If we collect our estimates, we arrive to the following intermediate result:

 \begin{thm}\label{T:reduced}
 Under the conditions of Theorem \ref{T:main}, the following asymptotics holds for the reduced Green's function:
\begin{multline*}
\Go(x,y) = \frac{1}{(2\pi)^d} \int_B   e^{i(x-y)\cdot k}
\frac{\overline{\vp(k,y)}\vp(k,x)} {\|\vp(k,\cdot)\|_{\T}^2 \lam(k)}\, dk \\
=  \frac{\Gamma(\frac d2 -1) e^{i(x-y)\cdot k_0}}
{2 \pi^{d/2}(\det A)^{1/2} } \left( \frac{1 }{|A^{-1/2}(x-y)|}\right)^{d-2}
\frac{\vp(k_0,x) \overline{\vp(k_0,y)}} {\|\vp(k_0)\|_{L^2(\T)}^2}  \big(1 +  O(|x-y|^{-1}) \big) \\+ O(|x-y|^{-d+1}).
\end{multline*}
\end{thm}

\subsection{The full Green's function asymptotics}
We now need to show that when descending from the true Green's function to the reduced one, we have not changed the principal term of the asymptotics.

Let us recall that the restriction to the reduced Green's function $G_0$ occurred in Section \ref{SS:principal} when we dropped  the operator term
$$
T:=\int^\bigoplus_{\T^*}T(k)d'k,
$$
where $T(k)=\nu(k)(L(k)Q(k))^{-1}+(1-\nu(k))(L(k))^{-1}$ is a smooth operator valued function of $k$. We will show now that the Schwartz kernel of this operator in $\R^n$
decays sufficiently fast to be included into the error term $r(x,y)$ in (\ref{E:expansion}).

To do so, we will need the following auxiliary result:
\begin{lemma}\label{L:local}
         For any compactly supported functions $\phi,\theta\in C^\infty_{0}(\R^d)$,
    the operator norm in $L^2(\R^d)$ of the operator $\phi_{\g'} T \theta_\g$, where $g_\g(x):=g(x+\g)$, satisfies  the estimate
    \begin{equation}\label{E:kern_decay}
        \|\phi_{\g'} T \theta_\g\|=O(|\g-\g'|^{-N}) \quad \text{ for any } N>0  \mbox{ when } |\g-\g'|\to\infty.
    \end{equation}
\end{lemma}

\begin{proof}
By using unity partitions and shifts, one can assume without loss of generality
that $\supp\phi\subset W$ and $\supp\theta\subset W$. Then the Floquet direct integral decomposition of the operator $\phi_{\g'} T \theta_\g$ is (up to an absolute value $1$ exponential factor)
$$
\phi(x)\int^\bigoplus_{\T^*} e^{ik\cdot (\g-\g')}T(k) \theta (x)d'k,
$$
where the functions $\phi$ and $\theta$ are considered as the multiplication operators by them in $L^2(\T^*)$. Then the estimate (\ref{E:kern_decay}) is just the decay estimate for Fourier coefficients of an infinite differentiable vector (operator) valued function on the torus (in this case, $\phi T(k) \theta$).
\end{proof}

The operator norm estimate (\ref{E:kern_decay}) is weaker than the pointwise estimate of the Schwartz kernel $r$ of the remainder in Theorem \ref{T:main}. However, one can bootstrap it from (\ref{E:kern_decay}) to these pointwise estimates. This can be done in different ways. One would be to use Schauder estimates, while another is to apply the similar reasoning to the Floquet expansion of the operator $T$ in a sufficiently smooth Sobolev space $H^s(\R^d)$ instead of $L^2$, which would give corresponding smooth norms estimates on the Schwartz kernel of $\phi_{\g'} T \theta_\g$.

This finishes the proof of Theorem \ref{T:main}.

\section{Final remarks}\label{S:remarks}
\begin{itemize}
  \item In the case when the weaker condition A$3^\prime$ is satisfied, one needs to add the asymptotics coming from all non-degenerate isolated extrema.
  \item It is easy to reformulate (by spectral shift and, if necessary, changing the sign of the operator) the asymptotics without assuming that $\lambda_0=0$ and that this is the upper, rather than lower bank of the gap.
  \item One might wonder whether one truly needs to have the additive error $r$ in (\ref{E:expansion}) of Theorem \ref{T:main}, not only the multiplicative one. This seems to be indeed necessary, unless one is dealing with the bottom of the spectrum. The reason is that the eigenfunctions $\phi$ in (\ref{E:expansion}) will have zeros, and thus the additive error $r$ cannot be pulled in as a part of the multiplicative one.
  \item Due to absence of restrictions on the coefficients of the operator $L$, the result holds for an arbitrary lattice of periods, not necessarily $\Z^d$.
  \item The result carries over without any essential changes in the arguments to periodic second order elliptic operators on abelian coverings of compact manifolds (as in \cite{KuchPin1,KuchPin2,Kuch_Molch,LinPinch}).
  \item The smoothness conditions on the coefficients of $L$ are exceedingly severe. One can easily follow the proof to detect that a finite smoothness would suffice. We decided not to pursue this line of study in this article.
\item As we have already briefly mentioned, the asymptotics we get has the spirit of an unconventional  analog of ``homogenization,'' which, unlike the classical homogenization, occurs at the internal gap edges, rather than at the bottom of the spectrum. One can find other results that can be interpreted as incarnations of such homogenization in \cite{KuchPin1,KuchPin2,Birman-homog,BirmSu8,BS3,KotSun,KotSun2}.
\end{itemize}
\section*{Acknowledgments}
P.K. expresses his thanks to V. Papanicolaou, Y. Pinchover, and T.~Tsuchida for useful discussions and suggestions. The work of P.K. was supported in part by IAMCS through the KAUST Award No. KUS-C1-016-04. The work of A.R. was supported in part by the NSF grant DMS-0855822. The authors express their gratitude to these institutions for the support.

\end{document}